\documentclass[a4paper, 11pt, fleqn, leqno]{article}
\RequirePackage{fullpage}

\usepackage{amsthm,amsmath,amssymb}
\usepackage{mathtools}
\usepackage{subcaption}
\usepackage[dvipsnames]{xcolor}
\usepackage[pdftex,breaklinks,colorlinks,
    citecolor={BlueViolet}, linkcolor={Blue},urlcolor=Maroon]{hyperref}
\usepackage{tikz,pgfkeys}
\usetikzlibrary{arrows.meta}
\usetikzlibrary{quotes}
\usetikzlibrary {shapes.geometric}
\usepackage{charter,eulervm}\usepackage{enumitem} 
\usepackage[final,expansion=alltext,protrusion=true]{microtype}
\usepackage{marvosym} 

\theoremstyle{plain} \newtheorem{theorem}{Theorem}[section]
\newtheorem{lemma}[theorem]{Lemma}
\newtheorem{proposition}[theorem]{Proposition}

\tikzset{
  max clique/.style  = {rounded corners, draw= gray, fill = white, inner sep=4pt, minimum width=1.5pt},
  empty vertex/.style  = {circle, draw, fill = white, inner sep=1.5pt, minimum width=1.5pt}  
}

\title{Cluster Vertex Deletion on Chordal Graphs}
\author{
  Yixin Cao\thanks{Department of Computing, Hong Kong Polytechnic University, Hong Kong, China. \texttt{yixin.cao@polyu.edu.hk}. }
  \and
  Peng Li\thanks{School of Mathematics and Computer Science, Hanjiang Normal University, Shiyan, China. lipengsjtu@163.com}
}
\date{}

\begin{document}
\maketitle

\begin{abstract}
  We present a polynomial-time algorithm for the cluster vertex deletion problem on chordal graphs, resolving an open question posed in different contexts by Cao et al.~[Theoretical Computer Science, 2018], Aprile et al.~[Mathematical Programming, 2023], Chakraborty et al.~[Discrete Applied Mathematics, 2024], and Hsieh et al.~[Algorithmica, 2024].  We use dynamic programming over clique trees and reduce the computation of the optimal subproblem value to the minimization of a submodular set function.
\end{abstract}

\section{Introduction}

All graphs in this paper are undirected and simple.
For a graph~$G$, we denote its vertex set by~$V(G)$ and write $n = |V(G)|$.
For a vertex~$v\in V(G)$, we denote its (open) neighborhood by~$N(v)$ and its closed neighborhood by~$N[v] = N(v)\cup\{v\}$. For a vertex set~$X\subseteq V(G)$, we write~$N[X] = \bigcup_{v\in X} N[v]$ and~$N(X) = N[X]\setminus X$.
The graph comes with a weight function~$w : V(G) \to \mathbb{Q}$.  The weight of a vertex set~$X \subseteq V(G)$ is~$w(X) = \sum_{v \in X} w(v)$.
A \emph{clique} is a set of vertices that are pairwise adjacent.
We denote by~$\omega(G)$ the weight of a maximum-weight clique in~$G$.

A \emph{cluster graph} is a graph whose components are all cliques.
In the \emph{cluster vertex deletion} problem, we are given a vertex-weighted graph and asked to delete a minimum-weight set of vertices so that the resulting graph is a cluster graph.
Cluster graphs naturally arise in graph-based data clustering, where one seeks to ``clean'' noisy data by removing a small number of outliers.
The cluster vertex deletion problem is NP-hard~\cite{lewis-80-node-deletion-np}.
Since a graph is a cluster graph if and only if it has no~$P_{3}$ (an induced path on three vertices), cluster vertex deletion can also be seen as a hitting set problem in a 3-uniform hypergraph, where each hyperedge corresponds to an induced~$P_{3}$.  This renders cluster vertex deletion a canonical benchmark for algorithmic techniques in parameterized complexity~\cite{tian-25-cvd}, kernelization~\cite{bessy-23-rainbow-matching}, and approximation algorithms~\cite{aprile-23-cvd}.
The minimum vertex-deletion weight on a graph $G$ is called the \emph{distance to cluster} of $G$ and is widely used as a structural parameter~\cite{doucha-12-cluster-vertex-deletion-as-parameter, goyal-25-dominating-set}.
It is one of the few nontrivial structural parameters that remain meaningful on dense graphs, in contrast with the abundance of parameters for sparse graphs.\footnote{See also \url{https://www.graphclasses.org/classes/par_29.html}.}

The forbidden-subgraph characterization of cluster graphs, together with the local ratio technique~\cite{bar-yehuda-04-local-ratio}, immediately yields a 3-approximation algorithm for cluster vertex deletion.
Aprile et al.~\cite{aprile-23-cvd} improved the ratio to 2 by reducing the problem to instances in which for all vertices~$v$, the subgraph induced by~$N(v)$ is free of any induced cycle of length at least four or a $2 P_{3}$ (two disjoint induced $P_{3}$'s with no edge between them).
A graph is \emph{chordal} if it does not contain an induced cycle of length at least four.
This motivated them to pose an open question on the complexity of cluster vertex deletion on $2 P_{3}$-free chordal graphs; they suggested in particular that the problem might be hard on general chordal graphs.

Chordal graphs form a central graph class in which many vertex-deletion problems become tractable through clique-tree decompositions, and thus they provide a natural boundary case for structural algorithmic questions.  Independently, Cao et al.~\cite{cao-18-vertex-deletion-chordal} raised the complexity of cluster vertex deletion on chordal graphs in their systematic study of vertex-deletion problems between subclasses of chordal graphs.  They showed that the problem is polynomial-time solvable on two important subclasses, namely, split graphs, whose vertex set can be partitioned into a clique and an independent set, and interval graphs, which are defined in the next section.  However, the case of chordal graphs remained one of the three open problems posed in~\cite{cao-18-vertex-deletion-chordal}.  More recently, Chakraborty et al.~\cite{chakraborty-24} generalized the split-graph algorithm of Cao et al.~\cite{cao-18-vertex-deletion-chordal} to well-partitioned chordal graphs, a class lying strictly between split graphs and chordal graphs, and reiterated the same open question.  This graph class was introduced by Ahn et al.~\cite{ahn-22} to better understand problems that are tractable on split graphs but hard on chordal graphs.

Cluster vertex deletion also appears naturally in the study of vertex deletion to $d$-claw-free graphs.
Bonomo{-}Braberman et al.~\cite{bonomo-24-claw-vertex-deletion} and Hsieh et al.~\cite{hsieh-24-claw-vertex-deletion} investigated the problem of deleting vertices to destroy all induced $K_{1,d}$'s (so-called $d$-claws).
Note that $1$- and $2$-claws correspond to edges and $P_{3}$'s, respectively.
Since the complexity of vertex cover (the $d=1$ case) is well understood on most natural graph classes~\cite{alekseev-82-independent-set}, cluster vertex deletion (the $d=2$ case) can be viewed as the simplest nontrivial problem in this framework.
In contrast, 3-claw vertex deletion is already NP-hard on split graphs~\cite{bonomo-24-claw-vertex-deletion}, and remains NP-hard even on $4$-claw-free split graphs~\cite{hsieh-24-claw-vertex-deletion}.

\paragraph{Our contribution.}
We resolve this open question with a dynamic-programming algorithm for cluster vertex deletion on chordal graphs.

\begin{theorem}\label{thm:alg}
There is a polynomial-time algorithm for cluster vertex deletion on chordal graphs. \end{theorem}

It will be more convenient to work with the equivalent \emph{maximum-weight induced cluster subgraph} problem.
A vertex set~$S \subseteq V(G)$ is called a \emph{solution} (of {maximum-weight induced cluster subgraph}) if $G[S]$ is a cluster graph, in which case each component of~$G[S]$ is a \emph{cluster} of the solution.
The goal is to find a solution~$S$ maximizing the total weight, and we write $\psi(G)$ for its weight---then~$V(G)\setminus S$ is a minimum-weight deletion set for the cluster vertex deletion problem.
We will consistently use this formulation in the rest of the paper.

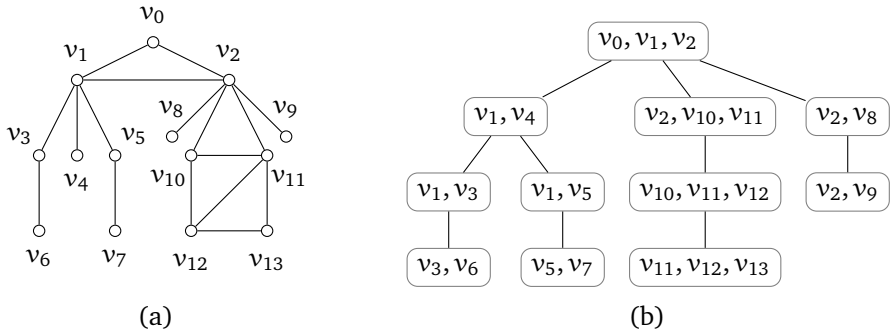
\begin{figure}[h]
  \centering\small
  \begin{subfigure}[b]{0.3\linewidth}
    \centering
    \begin{tikzpicture}[every node/.style={empty vertex}, scale=1]
      \foreach[count = \i from 0] \v/\l in {(0, 4.5)/, (-1, 4)/, (1, 4)/,
        (-1.5, 3)/above left, (-1, 3)/{below}, (-.5, 3)/above right,
        (-1.5, 2)/below, (-.5, 2)/below,
        (.25, 3.25)/, (1.75, 3.25)/,
        (.5, 3)/below left, (1.5, 3)/below right,
        (.5, 2)/below, (1.5, 2)/below}
      \node["$v_{\i}$" \l] (v\i) at \v {};
\draw[] (v6) -- (v3) -- (v1);
      \draw[] (v7) -- (v5) -- (v1);
      \draw[] (v4) -- (v1) -- (v0);
      \draw[] (v1) -- (v2) -- (v0);
      \draw[] (v8) -- (v2) -- (v9);
      \draw[] (v12) -- (v10) -- (v2);
      \draw[] (v13) -- (v11) -- (v2);
      \draw[] (v10) -- (v11) -- (v12) -- (v13);
\end{tikzpicture}
    \caption{}\end{subfigure}
  \begin{subfigure}[b]{0.5\linewidth}
    \centering
    \begin{tikzpicture}[every node/.style={max clique}, yscale=.5, xscale=.75]
      \small
      \foreach[] \p/\c/\x in {(4.5, 6)/{$v_{0}, v_{1}, v_{2}$}/012, (2, 4)/{$v_{1}, v_{4}$}/14,
        (1, 2)/{$v_{1}, v_{3}$}/13, (3, 2)/{$v_{1}, v_{5}$}/15,
        (1, 0)/{$v_{3}, v_{6}$}/36, (3, 0)/{$v_{5}, v_{7}$}/57,
        (5.5, 4)/{$v_{2}, v_{10}, v_{11}$}/2ab, (5.5, 2)/{$v_{10}, v_{11}, v_{12}$}/abc, (5.5, 0)/{$v_{11}, v_{12}, v_{13}$}/bcd, 
        (8, 4)/{$v_{2}, v_{8}$}/28, (8, 2)/{$v_{2}, v_{9}$}/29} {
        \node (\x) at \p {\c};
}
      \draw[] (36) -- (13) -- (14);
      \draw[] (57) -- (15) -- (14);
      \draw[] (29) -- (28) -- (012);
      \draw[] (bcd) -- (abc) -- (2ab);
      \draw[] (14) -- (012) -- (2ab);
    \end{tikzpicture}
    \caption{}\end{subfigure}
  \caption{(a) A chordal graph with 11 maximal cliques; (b) one of its clique trees.}
  \label{fig:clique-tree}
\end{figure}

Chordal graphs admit a well-known structural decomposition via clique trees.  See Figure~\ref{fig:clique-tree} for an example, where each node of the tree in Figure~\ref{fig:clique-tree}(b) corresponds to a maximal clique of the graph in Figure~\ref{fig:clique-tree}(a).
A chordal graph~$G$ has only linearly many maximal cliques, and these cliques can be arranged as a tree~$T$ such that, for every vertex~$v \in V(G)$, the nodes of~$T$ corresponding to maximal cliques containing~$v$ form a subtree~\cite{dirac-61-chordal-graphs}.
Such a tree~$T$ is called a \emph{clique tree} of~$G$ and can be constructed in linear time.
Dynamic programming on a clique tree, processed in a bottom-up fashion, is the standard algorithmic technique on chordal graphs, which is the basis of dynamic programming on tree decompositions of general graphs.

Our definition of subproblems of the dynamic programming on a clique tree is different from most algorithms of the same approach.
For each node~$K$ (i.e., maximal clique of~$G$) of the rooted clique tree~$T$, we consider the subproblem on the subgraph~$G_{K}$ induced by the vertices that belong \textit{only} to maximal cliques in the subtree rooted at~$K$; vertices that also appear in the parent of~$K$ (if any) are excluded.
If an optimal solution to this subproblem does not intersect $K$, it can be obtained by combining optimal solutions of the subproblems associated with the children of~$K$.
(This explains our ``unusual'' subproblem definition: we must exclude vertices in~$K$ from these subproblems.)
Consequently, the main difficulty is to handle optimal solutions that intersect $K$.  

Any such solution can be readily read from the dynamic programming table once we know the unique cluster~$C$ that intersects $K$.
This cluster must be a subset of some maximal clique~$Q$ in the subtree of~$K$.
We ``guess'' such a~$Q$ and a vertex~$v \in K \cap Q$, and then choose the remaining vertices of~$C$ as a set~$X \subseteq Q\setminus \{v\}$ that maximizes~$\psi( G_{K} - N[ X\cup \{v\} ] ) + w( X\cup \{v\} )$: 
this objective is exactly the contribution of the cluster $C = X\cup \{v\}$ plus the best solution in the remaining components after deleting the cluster and its neighbors.
The key structural ingredient is that, for fixed $K$, $Q$, and~$v$, the objective function on~$X$ is supermodular: a function~$f : 2^{U} \to \mathbb{Q}$ is \emph{supermodular} if for all $X, Y \subseteq U$,
\[
  f(X\cup Y) + f(X\cap Y) \ge f(X) + f(Y).
\]
Thus, we can maximize it in polynomial time using Jiang's oracle-based algorithm~\cite{jiang-22-minimize-convex-functions}.  The value of~$\psi( G_{K} - N[ X\cup \{v\} ] )$ is obtained from the dynamic programming table, since~$G_{K} - N[ X\cup \{v\} ]$ decomposes into subgraphs corresponding to subtrees in the clique tree.  Combining these cases for all nodes~$K$ yields a polynomial-time procedure for updating the dynamic programming table and thereby proves Theorem~\ref{thm:alg}.

We will prove the underlying structural statement in a more general form that might be of independent interest.
Fix a clique $K$ of a weighted chordal graph~$H$.  For any~$X\subseteq K$, let~$g(X)$ be the weight of a maximum-weight solution of~$H$ in which $X$ is a cluster.
\begin{theorem}\label{thm:submodular}
  The function $g : 2^{K} \to \mathbb{Q}$ is supermodular.
\end{theorem}
The proof presented in Section~\ref{sec:supermodularity} can be read independently from the others.

\section{Interval graphs}
This section can be skipped without loss of continuity; some reader may find it helpful in understanding our main algorithm, which can be viewed as a generalization of this section.

A graph is an \emph{interval graph} if its vertices can be represented by intervals on the real line so that two vertices are adjacent if and only if their corresponding intervals intersect.
Every interval graph is chordal and admits a clique tree that is a path, called a \emph{clique path}.
Cao et al.~\cite[Theorem 3.6]{cao-18-vertex-deletion-chordal} solved cluster vertex deletion on interval graphs using dynamic programming.
We present a simpler dynamic program for this problem.
Let $G$ be an interval graph and let $\langle K_{1}, K_{2}, \ldots, K_{\ell} \rangle$ be a clique path of $G$.
For convenience, we introduce a dummy set $K_{\ell+1}$ and set $K_{\ell+1} = \emptyset$.
For $p = 1, 2, \ldots, \ell$, define
\[
  G_{p} = G\left[ \left(\bigcup_{i=1}^{p} K_{i} \right) \setminus K_{p+1} \right].
\]
Note that $G_{\ell} = G$.
We will compute $\psi(G_{p})$ for all $p$ by dynamic programming.
Let~$S_{p}$ be a maximum-weight solution of~$G_{p}$, with~$t$ clusters~$C_{1}, C_{2}, \ldots, C_{t}$.
The proof of the following simple observation is omitted because it is a special case of Proposition~\ref{lem:chordal-clusters}, to be proved in the next section.
\begin{proposition}\label{lem:clusters}
  For each $i = 1, \ldots, t$, there exists an index $j \in \{1, \ldots, p\}$ such that
$C_{i} \subseteq K_{j} \setminus K_{p+1}.$
\end{proposition}

The index $j$ in Proposition~\ref{lem:clusters} does not need to be unique.
If $C_{i} \subseteq K_{q} \setminus K_{p+1}$ and $C_{i} \subseteq K_{r} \setminus K_{p+1}$ for indices $q < r$, then by the definition of clique paths we have
\[
  C_{i} \subseteq K_{j} \setminus K_{p+1}
  \quad \text{for all } j = q, \ldots, r.
\]

Among the clusters $C_{1}, \ldots, C_{t}$ of $S_{p}$, choose $C_{t}$ to be the \emph{rightmost} one, in the sense that every vertex of $S_{p} \setminus C_{t}$ lies in some $K_{i}$ with $i$ strictly smaller than the smallest index of a clique intersecting $C_{t}$.
Let
\[
  j = \min \{ i \in \{1, \ldots, p\} \mid C_{t} \cap K_{i} \neq \emptyset \}.
\]
If $j = 1$, then $t = 1$, i.e., $S_{p}$ consists of a single cluster, and hence
\[
  \psi(G_{p}) = w(C_{t}) = \omega(G_{p}).
\]
Thus we may assume $j > 1$.
Note that~$S_{p} \setminus C_{t}$ is an optimal solution of $G_{j-1}$.

By the minimality of $j$, we have $C_{t} \cap K_{j-1} = \emptyset$.
Consequently, $C_{t} \cap K_{j}$ is disjoint from $K_{j-1}$.  By the definition of clique paths, we have~$C_{t}\subseteq (K_{j} \cup \cdots \cup K_{p}) \setminus (K_{j-1} \cup K_{p+1})$.
In summary, $C_{t}$ is a maximum-weight clique of
\[
  G\left[ (K_{j} \cup \cdots \cup K_{p}) \setminus (K_{j-1} \cup K_{p+1}) \right],
\]
because any clique of this subgraph makes a solution of~$G_{p}$ together with~$S_{p} \setminus C_{t}$.
Therefore,
\[
  \psi(G_{p}) = \psi(G_{j-1}) + \omega\left( G\left[ (K_{j} \cup \cdots \cup K_{p}) \setminus (K_{j-1} \cup K_{p+1}) \right] \right).
\]
Rewriting this in a fully dynamic-programming form, we obtain
\begin{equation}
  \label{eq:dp-cvd-interval}
  \psi(G_p)
  = \max_{i = 0}^{p - 1} \left\{ \psi(G_i) +
    \max_{j = i+1}^{p} \omega \left( G\left[ K_{j} \setminus (K_{i} \cup K_{p+1}) \right] \right)
  \right\},
\end{equation}
where, for convenience, we interpret $K_{0} = \emptyset$ and $G_{0}$ as the ``null graph'' with $\psi(G_{0}) = 0$.
Note that~$\omega \left( G\left[ K_{j} \setminus (K_{i} \cup K_{p+1}) \right] \right)$ refers to the maximum-weight clique among vertices that appear in~$K_j$ but do not belong to any clique~$K_{r}$ with~$r\le i$ or~$r\ge p+1$, i.e., the ``rightmost'' cluster with respect to the clique path.

It is straightforward to implement the recurrence~\eqref{eq:dp-cvd-interval} as a dynamic-programming algorithm.
Since there are~$O(n)$ maximal cliques and each inner maximization over~$\omega(\cdot)$ can be evaluated in~$O(n)$ time with appropriate preprocessing, the overall running time is~$O(n^{2})$.

\section{Chordal graphs}

We now describe the main algorithm stated in Theorem~\ref{thm:alg}, assuming Theorem~\ref{thm:submodular}, which will be proved in the next section.
Let~$G$ be a chordal graph and let~$T$ be a clique tree of~$G$.  
Choose an arbitrary node~$R$ and root~$T$ at~$R$.
For each node~$K \neq R$, let~$P(K)$ denote the parent of~$K$; for convenience, set~$P(R) = \emptyset$.
For each node~$K$, let~$T_K$ denote the subtree of~$T$ rooted at~$K$, and \[
  G_{K} = G\left[ \bigcup_{Q\in T_K} Q \setminus P(K) \right].
\]
Note that~$G_R = G$.
We use dynamic programming to compute~$\psi(G_{K})$ in a bottom-up manner.
If~$K$ is a leaf, then~$G_K$ is a complete graph with
\[
  \psi(G_{K}) = \omega(G_{K}) = w( V(G_K) ) = w( K \setminus P(K) ).
\]
Now consider a non-leaf node~$K$.
Let~$S_{K}$ be an optimal solution of~$G_{K}$, with clusters~$C_{1}, C_{2}, \ldots, C_{t}$.
\begin{proposition}\label{lem:chordal-clusters}
  For each~$i = 1, \ldots, t$, there exists a node~$Q$ in~$T_{K}$ such that~$C_{i}\subseteq Q\setminus P(K)$.
\end{proposition}
\begin{proof}
  Since $C_{i}$ is a clique, it is a subset of some maximal clique $Q$ of $G$.
  To prove that~$Q\in T_K$, suppose for a contradiction that~$Q$ lies outside $T_K$.  Let~$v$ be any vertex of~$C_{i}$.
  Because $v\in C_{i} \subseteq V(G_{K})$, the vertex~$v$ belongs to some maximal clique in $T_K$.
  On the other hand, $v\in Q$ and~$Q\not\in T_K$.
  By the connectedness property of the set of maximal cliques containing~$v$, all such cliques form a connected subtree of~$T$; this subtree cannot simultaneously intersect both~$T_{K}$ and a node outside~$T_{K}$ without passing through~$K$ and its parent.
  Hence,~$v$ must be contained in both~$K$ and~$P(K)$, that is, $v \in K\cap P(K)$.  But such vertices are excluded from~$V(G_{K})$, contradicting~$v\in C_{i} \subseteq V(G_{K})$.

  Therefore,~$Q$ must lie in $T_K$.  Since $V(G_{K})$ excludes all vertices in $K\cap P(K)$, we also have $C_{i}\cap P(K) = \emptyset$.  Thus,~$C_{i}\subseteq Q\setminus P(K)$, as claimed.
\end{proof}

If the solution~$S_{K}$ is disjoint from~$K$, then for each child~$Q$ of~$K$, the set~$S_{K}\cap V(G_Q)$ is an optimal solution of the subgraph~$G_Q$, and hence
\[
  \psi(G_K) = \sum_{Q: P(Q) = K} \psi( G_Q ).
\]
In the other case,~$S_{K}\cap K\ne \emptyset$.
Without loss of generality, assume that~$C_{t}\cap K\ne \emptyset$.
Note that~$S_{K}\setminus C_{t}$ is an optimal solution of~$G_{K} - N[C_{t}]$, the subgraph of~$G_{K}$ obtained by removing all vertices in~$C_{t}$ as well as their neighbors.

\begin{lemma}\label{lem:subproblem}
  If $C_{t}\cap K\ne \emptyset$, then every component of $G_{K} - N[C_{t}]$ is a component of $G_{Q}$ for some node $Q$ in $T_{K}$ different from~$K$.
\end{lemma}
\begin{proof}
  Let $A$ be a component of $G_{K} - N[C_{t}]$, and let $Q$ be the lowest node in~$T_{K}$ (i.e., farthest from the root~$R$) such that $A\subseteq V(G_{Q})$.  Such a node exists because $A\subseteq V(G_{R}) = V(G)$.

  We first show that $Q\cap C_{t} = \emptyset$ (and hence $Q\neq K$).
  Suppose, for a contradiction, that there exists a vertex $x\in Q\cap C_{t}$.
  Since $Q$ is a clique and $x\in Q$, we have $Q\subseteq N[x]\subseteq N[C_{t}]$, so no vertex of $Q$ lies in $G_{K} - N[C_{t}]$.
  Consider the children of $Q$ in the clique tree. For any two distinct children $X$ and $Y$ of $Q$, there is no edge of $G$ with one endpoint in $V(G_{X})$ and the other in $V(G_{Y})$: otherwise, the ends of that edge would be contained in some maximal clique, but there is no such a node in~$T$ corresponding to this clique.
  Since $A$ is a component of $G_{K} - N[C_{t}]$ and $A\subseteq V(G_{Q})\setminus Q$, it follows that $A$ must be entirely contained in $V(G_{Z})$ for some child $Z$ of $Q$.  This contradicts the choice of $Q$ as the lowest node with $A\subseteq V(G_{Q})$.
  Therefore $Q\cap C_{t} = \emptyset$, and in particular $Q\neq K$.

Note that $A\subseteq V(G_{Q})$ and $A$ is a component of $G_{K} - N[C_{t}]$, so $A$ is a component of $G_{Q} - N[C_{t}]$.  Since $Q\cap C_{t} = \emptyset$, it follows that $V(G_{Q}) \cap N[C_{t}] = \emptyset$.  Thus, $A$ is a component of $G_{Q}$.
\end{proof}

By Lemma~\ref{lem:subproblem}, the value $\psi( G_{K} - N[C_{t}] )$ can be obtained directly from the dynamic programming table.
The problem of computing $\psi(G_{K})$ therefore reduces to
\begin{quote}
  Find a clique $C^*\subseteq V(G_{K})$ with $C^*\cap K \ne \emptyset$ that maximizes $w( C^* ) + \psi( G_{K} - N[C^*] )$.
\end{quote}
Note that $C^*$ is not necessarily a subset of $K$.
By Proposition~\ref{lem:chordal-clusters}, we have~$C^*\subseteq Q$ for some node~$Q$ in~$T_K$.
For each node~$Q$ in~$T_K$, we may compute a subset~$C^*_{Q}\subseteq Q$ such that~$C^*_{Q}\cap K \ne \emptyset$ and
\[
w(C^*_{Q}) + \psi( G_{K} - N[C^*_{Q}] )
\]
is maximized; it then suffices to take the maximum of these values over all nodes~$Q$ in~$T_{K}$.

However, a subtle but crucial issue is that we cannot apply Theorem~\ref{thm:submodular} directly to find~$C^*_{Q}$.
The algorithm of minimizing a submodular function requires a polynomial-time oracle of evaluating any set~\cite{jiang-22-minimize-convex-functions}.
The intersection of two subsets~$X, Y\subseteq Q$ that both intersect~$K$ may be disjoint from~$K$.
If~$X\cap Y$ is disjoint from~$K$, then we are not able to evaluate~$\psi( G_{K} - N[X\cap Y] )$ from the dynamic programming table.
We proceed as follows.
For each~$v\in K\cap Q$, we compute a subset~$C^*_{Q, v} \subseteq Q$ containing~$v$ that maximizes
\[
\psi\left( G_{K} - N \left[ C^*_{Q, v}  \right] \right) + w\left( C^*_{Q, v} \right).
\]
Equivalently, we work with the function~$f_{v}: 2^{Q\setminus \{v\}}\to \mathbb{Q}$ defined by
\[
  f_{v}(X) = g( X\cup \{v\} ) = \psi( G_{K} - N[ X\cup \{v\} ] ) + w( X\cup \{v\} ),
\]
where~$g$ is the function from Theorem~\ref{thm:submodular}, defined on the chordal graph~$G_{K}$ with respect to clique~$Q\setminus P(K)$.
Let~$H = G_{K}$ and let $K'=Q\setminus P(K)$.  For each $X\subseteq K'$, let $g(X)$ be the weight of an optimal solution of $H$ in which $X$ appears as a cluster.  Then Theorem~\ref{thm:submodular} says that $g$ is supermodular on $2^{K'}$.
The function~$f_{v}$ is supermodular because
\[
  f_{v}(X\cup Y) + f_{v}(X\cap Y)
  = g( X\cup Y\cup \{v\} ) + g( X\cap Y\cup \{v\} )
  \ge
  g( X\cup \{v\} ) + g( Y\cup \{v\} )
  =   f_{v}(X) + f_{v}(Y),
\]
where the inequality follows from Theorem~\ref{thm:submodular}.

In summary,
for each clique~$K$ in the clique tree, we fix an optimal solution~$S_{K}$ of~$G_{K}$ and consider two cases.
\begin{itemize}
\item $S_{K}\cap K = \emptyset$: Then~$\psi(G_K)$ is the sum over children. \item $S_{K}\cap K \ne \emptyset$: Then we identify a maximal clique~$Q$ in the subtree containing the cluster of~$S_{K}$ that intersects~$K$, guess a vertex~$v \in K \cap Q$ that lies in that cluster, and then choose the remaining vertices of the cluster as~$X \subseteq Q \setminus \{ v \}$ by maximizing a supermodular function~$f_{v}(X)$.  The resulting value of~$f_{v}(X)$ gives~$\psi(G_K)$.
\end{itemize}
We are ready for the dynamic programming recurrence and use it to prove the main theorem:
\begin{equation}
  \label{eq:2}
  \psi( G_K ) = \max \left\{\quad\;
    \sum_{\mathclap{Q: P(Q) = K}} \psi(G_Q), \quad
    \max_{Q\in T_{K}}
    \max_{v\in Q\cap K\setminus P(K)}
    \max_{X\subseteq Q\setminus \{ v \}} f_{v}(X) 
\right\}.
\end{equation}

\begin{proof}[Proof of Theorem~\ref{thm:alg}]
We build a dynamic programming table to find a maximum-weight induced cluster subgraph of~$G$.  The vertices not in this subgraph form a solution to the cluster vertex deletion problem.
  Jiang~\cite[Theorem 1.7]{jiang-22-minimize-convex-functions} shows that a supermodular set function with a polynomial-time evaluation oracle can be maximized in polynomial time.  In particular, for each fixed maximal clique~$K$, $Q\in T_{K}$, and~$v\in K\cap Q$, we can maximize~$f_{v}$ over subsets of~$Q\setminus\{ v \}$ using at most~$O (|Q|^3)$ evaluations of~$f_{v}$.
  In our setting, each evaluation is done by computing~$\psi( G_{K} - N[X\cup \{ v \} ] )$ from the dynamic programming table in~$O(n)$ time: by Lemma~\ref{lem:subproblem}, every component of $G_{K} - N[X\cup \{ v \} ]$ is the vertex set of $G_{Z}$ for some node $Z\in T_{K}\setminus \{K\}$.  Moreover, if two such nodes lie in the same rooted subtree, only the highest one is relevant, because its graph $G_{Z}$ already contains all lower descendants.  Therefore, $G_{K} - N[X\cup \{ v \} ]$ is the disjoint union of the graphs $G_{Z}$ over the roots $Z$ of the surviving maximal subtrees of~$T_{K}$, and its optimal value is obtained by summing the precomputed values $\psi(G_{Z})$.

  For a fixed~$K$, there are~$O(n)$ choices of~$Q \in T_K$ and, for each~$Q$, at most~$O(n)$ choices of~$v$.  For each pair~$(Q, v)$, Jiang's algorithm makes $O(|Q|^3) = O(n^3)$ oracle calls, each in~$O(n)$ time, yielding $O(n^5)$ per~$(Q, v)$ and thus $O(n^6)$ per~$K$.  There are $O(n)$ nodes~$K$, and this gives the claimed $O(n^7)$ bound.
\end{proof}
We remark that the $O(n^7)$ running time can likely be improved with a more careful analysis of our algorithm.

\section{Theorem~\ref{thm:submodular}: Supermodularity of the function $g$}
\label{sec:supermodularity}

In this section we explain the structural reason why the function $g$ is supermodular.
Before the proof, we present a high-level overview of the main ideas.

Recall that $K$ is a fixed clique of a weighted chordal graph~$H$, and for a subset $X \subseteq K$, the value $g(X)$ is defined as the weight of an optimal solution of $H$ in which $X$ appears as a cluster. Our goal is to prove that for any two subsets $A_1, A_2 \subseteq K$,
\begin{equation}
  \label{eq:supermodular}
  g(A_1 \cup A_2) + g(A_1 \cap A_2) \ge g(A_1) + g(A_2).
\end{equation}

Suppose neither~$A_1$ nor~$A_2$ is contained in the other, and for $i=1,2,$ let $S_i$ be an optimal solution of $H$ in which $A_i$ is a
cluster, so $g(A_i) = w(S_i)$. We may assume that $V(H) = S_1 \cup S_2$, since
deleting vertices outside this union can only decrease the left-hand side of
\eqref{eq:supermodular} while leaving the right-hand side unchanged.

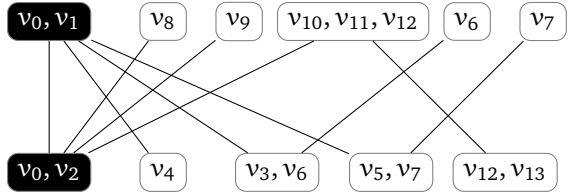
\begin{figure}[h]
  \centering\small
  \begin{subfigure}[b]{0.4\linewidth}
    \centering
    \begin{tikzpicture}[every node/.style={max clique}, scale=1]
      \foreach[count=\i from 2] \v/\x in {{$v_{8}$}/1.5, {$v_{9}$}/2.5, {$v_{10}, v_{11}, v_{12}$}/4, {$v_{6}$}/5.5, {$v_{7}$}/6.5}
      \node (x\i) at (\x, 2) {\v};
      \node[text=white, fill=black] (x1) at (0, 2) {$v_{0}, v_{1}$};
      \foreach[count=\i] \v in {{$v_{4}$}, {$v_{3}, v_{6}$}, {$v_{5}, v_{7}$}, {$v_{12}, v_{13}$}}
      \node (y\i) at ({\i*1.5}, 0) {\v};
      \node[text=white, fill=black] (y0) at (0, 0) {$v_{0}, v_{2}$};
      \draw[] (x1) -- (y0);
      \draw[] (x1) -- (y1);
      \draw[] (x1) -- (y2) -- (x5);
      \draw[] (x1) -- (y3) -- (x6);
      \draw[] (x2) -- (y0);
      \draw[] (x3) -- (y0);
      \draw[] (y0) -- (x4) -- (y4);
    \end{tikzpicture}
\end{subfigure}
  \caption{The bipartite graph~$B$ when the graph is in Figure~\ref{fig:clique-tree}, with unit weight, $K = \{v_{0}, v_{1}, v_{2}\}$, $A_{1} = \{v_{0}, v_{1}\}$, and $A_{2} = \{v_{0}, v_{2}\}$.
    Clusters corresponding to~$S_{1}$ are drawn at the top, and clusters corresponding to~$S_{2}$ at the bottom, where the nodes corresponding to~$A_{1}$ and~$A_{2}$ are highlighted.
    Edges represent pairs of clusters with at least one edge of~$H$ between them.}
  \label{fig:supermodular}
\end{figure}

The key idea is to study how the clusters of $S_1$ and $S_2$ interact. To this
end, we define an auxiliary graph $B$ whose nodes are all clusters of $S_1$
and of $S_2$, and we put an edge between two nodes if there is at least one
edge of $H$ with one endpoint in each of the corresponding clusters.
See Figure~\ref{fig:supermodular} for an illustration of~$B$ on the chordal graph from Figure~\ref{fig:clique-tree}.
By
construction, $B$ is bipartite, with one part consisting of clusters of $S_1$
and the other of clusters of $S_2$. The crucial structural property is that
$B$ is in fact a tree: the presence of a cycle in $B$ would give rise to an
induced cycle of length at least four in $H$, contradicting the chordality
encoded by the clique tree.

We now root $B$ at the edge $A_1 A_2$ (or the node obtained by merging $A_1$ and~$A_2$ into a single node).  Traversing $B$ from this root, we assign the
clusters on alternating levels to two new solutions $S_{\cup}$ and $S_{\cap}$:
the root node, corresponding to $A_1 \cup A_2$, is assigned to $S_{\cup}$, and
we add $A_1 \cap A_2$ as a cluster of $S_{\cap}$. Away from $A_1 \cup A_2$,
this procedure simply repartitions the clusters of $S_1$ and $S_2$ so that
each original cluster appears exactly once in either $S_{\cup}$ or
$S_{\cap}$. Consequently,
\[
  w(S_{\cup}) + w(S_{\cap}) = w(S_1) + w(S_2) = g(A_1) + g(A_2),
\]
and since $S_{\cup}$ and $S_{\cap}$ are feasible solutions with
$A_1 \cup A_2$ and $A_1 \cap A_2$, respectively, as clusters, they provide
lower bounds for $g(A_1 \cup A_2)$ and $g(A_1 \cap A_2)$. This yields
\eqref{eq:supermodular}, and thus establishes the supermodularity of $g$.

\begin{proof}[Proof of Theorem~\ref{thm:submodular}]
  Let $A_{1}, A_{2}$ be two subsets of~$K$.  We need to show~$g(A_{1}\cup A_{2}) + g(A_{1}\cap A_{2}) \ge g(A_{1}) + g(A_{2})$.
  Assume $A_{1} \not\subseteq A_{2}$ and~$A_{2} \not\subseteq A_{1}$; otherwise the equality holds trivially.
  For $i = 1, 2$, let $S_{i}$ be an optimal solution of $H$ in which $A_{i}$ is a cluster.
  We will construct two solutions $S_{\cup}$ and $S_{\cap}$ such that $A_{1}\cup A_{2}$ is a cluster of $S_{\cup}$ and $A_{1}\cap A_{2}$ is a cluster of $S_{\cap}$, and
  \[
    w \left(S_{\cup} \right) + w \left(S_{\cap} \right) = w \left(S_{1} \right) + w \left(S_{2} \right).
  \]

  We may assume that $V(H) = S_{1}\cup S_{2}$: deleting vertices outside $S_{1}\cup S_{2}$ does not change~$g(A_1)$ or~$g(A_2)$, and it cannot increase $g(A_{1}\cup A_{2})$ or $g(A_{1}\cap A_{2})$.
  Hence,~$K = A_{1}\cup A_{2}$.
  We also assume that $H$ is connected: since $A_{1}\cup A_{2}$ is a clique, it lies in a single component of~$H$, and for all other components we can use the same cluster subgraph in all solutions.  We have the following as a result.
\begin{equation}
    \label{eq:4}
    \text{No clique can simultaneously appear as a cluster in both $S_{1}$ and $S_{2}$.}\end{equation}
  Let~$C$ be any cluster in~$S_{1}$.
  Since~$H$ is connected, there is a vertex~$x\in V(H)\setminus C$ that has a neighbor in~$C$.
  Since~$V(H) = S_{1}\cup S_{2}$, we have~$x\in S_{2}$.
  Hence,~$C$ cannot be a cluster in~$S_{2}$.  The other direction is symmetric and omitted.

  We define an auxiliary bipartite graph~$B$ as follows.
  For each cluster in~$S_{1}$ and each cluster in~$S_{2}$, we introduce a node.
  For every two nodes~$C'$ and~$C''$, we add an edge if there is an edge in~$H$ between a vertex of~$C'$ and a vertex of~$C''$.
  The graph~$B$ is bipartite because there is no edge between distinct clusters of~$S_{1}$, nor between distinct clusters of~$S_{2}$.

  First, two nodes~$C'$ and~$C''$ in~$B$ are nonadjacent if and only if the corresponding clusters are disjoint and nonadjacent in~$H$.
  Indeed, if $C'\cap C''\ne \emptyset$, then there is an edge in~$H$ between a vertex of~$C'\cap C''$ and a vertex of~$(C'\cup C'')\setminus (C'\cap C'')$, which is nonempty by~\eqref{eq:4}, so the corresponding nodes are adjacent in~$B$.
  In particular, there is an edge between~$A_{1}$ and~$A_{2}$, and there is an edge whenever~$C' \subset C''$ or~$C'' \subset C'$.

  Second, we show that $B$ is connected.
  Suppose not, and let~$\mathcal{C}$ be a component of~$B$ not containing the node~$A_{1}$.
  Let $U = \bigcup_{C\in\mathcal{C}} C$.
  Note that $U \subseteq V(H)\setminus K$ because $A_{1}$ (and hence $A_{2}$) lie in a different component of~$B$.
  Since $H$ is connected, there is an edge of~$H$ between~$U$ and~$V(H)\setminus U$; let them be~$x$ and~$y$, respectively.
  By definition of~$U$, the vertex~$x$ belongs to some cluster~$C_{x}$ whose corresponding node lies in~$\mathcal{C}$.
  On the other hand, the vertex~$y$ is in some cluster~$C_{y}$ because~$V(H) = S_{1}\cup S_{2}$.
  There is an edge in~$B$ between~$C_{x}$ and~$C_{y}$.
  But then~$y$ must be in~$U$ as well, a contradiction.
  Hence~$B$ is connected.

  We now prove the key claim.\footnote{A similar argument shows that $B$ is in fact a tree.  The weaker statement~\eqref{eq:1} suffices for our purposes.}
\begin{equation}
    \label{eq:1}
    \text{The edge $A_{1}A_{2}$ is a bridge of~$B$ (its removal disconnects~$B$).}\end{equation}
  Suppose, to the contrary, that $A_{1}A_{2}$ lies in a cycle of~$B$.
  Take an induced cycle of~$B$ containing this edge, and write it as
  \[
    A_{1}, A_{2}, C_{1}, C_{2}, C_{3}, \ldots, C_{2\ell}.
  \]
  Since $B$ is bipartite, we have $C_{2i-1} \subseteq S_{1}$ and~$C_{2i} \subseteq S_{2}$ for~$i = 1, \ldots, \ell$.
  Along the cycle, we can choose a path in $H$ from a vertex in~$A_{2}\setminus A_{1}$ to a vertex in~$A_{1}\setminus A_{2}$, using one or two vertices from each cluster on the cycle, in such a way that all internal vertices lie in clusters other than $A_{1}$ and $A_{2}$.  
  Let~$C_{0} = A_{2}$ and~$C_{2\ell + 1} = A_{1}$.
  For~$i = 0, 1, \ldots, 2 \ell$, choose an edge between~$x_{i}^{+}\in C_{i}$ and~$x_{i+1}^{-}\in C_{i+1}$.
  The vertices~$x_{i}^{-}$ and~$x_{i}^{+}, i = 1, \ldots, 2 \ell,$ from the same cluster are either identical or adjacent.
Note that $x_{2\ell + 1}^{-}\in A_{1}\setminus A_{2}$ (because
$C_{2\ell}$, which contains~$x_{2\ell}^{+}$, and~$A_{2}$ are two distinct clusters of~$S_{2}$) and~$x_{0}^{+}\in A_{2}\setminus A_{1}$ by symmetry.
  Let
  \[
    P =
    \begin{cases}
      x_{1}^{+}, x_{2}^{-}, (x_{2}^{+}), \ldots, x_{2\ell}^{-}, (x_{2\ell}^{+}), x_{2\ell + 1}^{-} & \text{ if $x_{0}^{+}$ is adjacent to $x_{1}^{+}$},   \\
      x_{1}^{-}, x_{1}^{+}, x_{2}^{-}, (x_{2}^{+}), \ldots, x_{2\ell}^{-}, (x_{2\ell}^{+}), x_{2\ell + 1}^{-} & \text{ otherwise,}
    \end{cases}
  \]
  where $x_{i}^{+}, i = 2, \ldots, 2\ell,$ is used only if $x_{i}^{+} \neq x_{i}^{-}$.
In either case, $x_{0}^{+}$ is adjacent to the endpoints of~$P$, and by the choice of the vertices on the clusters of the induced cycle, $x_{0}^{+}$ has no other neighbors on~$P$.
The two ends of~$P$ are nonadjacent in~$H$ because they belong to different clusters of~$S_{1}$.
The path~$P$ is not necessarily induced, but we can always find an induced path connecting its ends using only vertices on~$P$.
Thus, $x_{0}^{+}$ together with the vertices on this induced path forms an induced cycle of length at least four, contradicting that $H$ is chordal.
  This proves that $A_{1}A_{2}$ is a bridge of~$B$.

  We now construct $S_{\cup}$ and~$S_{\cap}$.
  Consider the subgraph obtained from~$B$ after removing the bridge~$A_{1}A_{2}$ and for~$i = 1, 2,$ let $\mathcal{C}_{i}$ be the set of clusters corresponding to the connected component containing $A_{i}$.
  We define \begin{align*}
    S_{\cup} &= \left( \bigcup_{C \in \mathcal{C}_{1} ,\, C \subseteq S_{1}} C \right ) \cup \left( \bigcup_{C \in \mathcal{C}_{2} ,\, C \subseteq S_{2}} C \right ),
\\
    S_{\cap} &= \left( \bigcup_{C \in \mathcal{C}_{1} ,\, C \subseteq S_{2}} C \right ) \cup \left( \bigcup_{C \in \mathcal{C}_{2} ,\, C \subseteq S_{1}} C \right ) \cup ( A_{1}\cap A_{2} ).
\end{align*}
Note that $A_{1}, A_{2} \subseteq S_{\cup}$, and recall that $A_{1}\cup A_{2} = K$.

To show that $S_{\cup}$ and $S_{\cap}$ are valid solutions, it suffices to examine only the newly created distinguished clusters $A_1\cup A_2$ and $A_1\cap A_2$.  Indeed, every other cluster in $S_{\cup}$ or $S_{\cap}$ is inherited unchanged from one of the original solutions $S_1$ and $S_2$, and hence already induces a clique and has no neighbors inside the corresponding solution.  Therefore, we only need to verify that $A_1\cup A_2$ is a cluster of $S_{\cup}$ and that $A_1\cap A_2$ is a cluster of $S_{\cap}$.  By construction, $S_{\cup}$ contains no vertices in $N(A_1)\cap S_2$ and no vertices in $N(A_2)\cap S_1$.  Since $A_1\cup A_2=K$ and every vertex adjacent to $K$ that belongs to $S_1\cup S_2$ must lie in one of these two sets, it follows that $S_{\cup}\cap N(K)=\emptyset$.  Hence, $A_1\cup A_2$ is a cluster of $S_{\cup}$.  Moreover, from the assumption~$V(H)=S_1\cup S_2$ we obtain
\[
N(A_1\cap A_2)\cap S_i \subseteq A_i \qquad \text{for } i=1,2,
\]
because $A_i$ is a cluster of $S_i$.  Therefore,
\[
N(A_1\cap A_2)\subseteq A_1\cup A_2,
\]
and since $S_{\cap}$ contains no vertex of $(A_1\cup A_2)\setminus (A_1\cap A_2)$, the set $A_1\cap A_2$ has no neighbor in~$S_{\cap}$.  Thus, $A_1\cap A_2$ is a cluster of~$S_{\cap}$.

  Finally, observe that $S_{\cup}\setminus (A_{1}\cup A_{2})$ and~$S_{\cap}\setminus (A_{1}\cap A_{2})$ form a partition of~$(S_{1}\cup S_{2})\setminus (A_{1}\cup A_{2})$.
  Hence
\begin{align*}
    w \left(S_{\cup} \right) + w \left(S_{\cap} \right)
    =& w \left( (S_{1}\cup S_{2})\setminus (A_{1}\cup A_{2}) \right)
    + w(A_{1}\cup A_{2}) + w(A_{1}\cap A_{2})
    \\
    =&
       w \left(S_{1}\setminus A_{1} \right) + w \left(S_{2} \setminus A_{2} \right) + w(A_{1}) + w(A_{2})
       \\
    =& w \left(S_{1} \right) + w \left(S_{2} \right).
  \end{align*}
By the definition of the function~$g$, we have \[
    g(A_{1}\cup A_{2}) + g(A_{1}\cap A_{2})
    \;\ge\; w \left(S_{\cup} \right) + w \left(S_{\cap} \right)
    \;=\; w \left(S_{1} \right) + w \left(S_{2} \right)
    \;=\; g(A_{1}) + g(A_{2}),
  \]
  and hence $g$ is supermodular.
\end{proof}


\end{document}